\documentclass{llncs}

\usepackage{amsmath,amssymb,microtype,xcolor,graphicx,enumerate}
\usepackage[all]{xy}
\usepackage{booktabs}

\renewcommand*{\ge}{\geqslant}
\newcommand{\pref}{\succcurlyeq}

\newcommand*{\citep}[1]{\cite{#1}}
\newcommand*{\citet}[1]{\cite{#1}}

\newcommand{\lab}[1]{\colorbox{white}{\scriptsize #1}}
\newcommand{\labs}[1]{{\fboxsep 0.7pt \colorbox{white}{\scriptsize #1}}}

\let\doendproof\endproof
\renewcommand\endproof{~\hfill\qed\doendproof}

\begin{document}
\title{Precise Complexity of the Core in \\ Dichotomous and Additive Hedonic Games}
\titlerunning{Core in Dichotomous and Additive Hedonic Games}
\author{Dominik Peters}
\institute{Department of Computer Science \\
University of Oxford \\ 
Oxford, UK\\
\email{dominik.peters@cs.ox.ac.uk}}
\maketitle
\begin{abstract}
	Hedonic games provide a general model of coalition formation, in which a set of agents is partitioned into coalitions, with each agent having preferences over which other players are in her coalition. We prove that with additively separable preferences, it is $\Sigma_2^p$-complete to decide whether a core- or strict-core-stable partition exists, extending a result of Woeginger (2013). Our result holds even if valuations are symmetric and non-zero only for a constant number of other agents. We also establish $\Sigma_2^p$-completeness of deciding non-emptiness of the strict core for hedonic games with dichotomous preferences. Such results establish that the core is much less tractable than solution concepts such as individual stability.
\end{abstract}

\section{Introduction}

Suppose agents wish to form coalitions, perhaps to jointly achieve some task or common goal, with the payoff to an agent depending on the make-up of the coalition the agent is joining. In many situations, it makes sense to model agents' preferences to only depend on the \emph{identity} of the players in a group. Such games are called \emph{hedonic games}, because agents in a hedonic game can be seen as deriving pleasure from each other's presence.

An agent in a hedonic game specifies a preference ordering over all sets (\textit{coalitions}) of agents. An outcome of the game is a partition of the agent set into disjoint coalitions. A player prefers those partitions in which she is part of a preferred coalition. The main focus of the literature on hedonic games is studying outcomes that are \textit{stable} \citep{Banerjee2001,Bogomolnaia2002}.

Of the many notions of stability discussed in the literature, the most prominent is the concept of the \emph{core}. A partition $\pi$ of the agent set is core-stable if there is no non-empty set $S$ of agents all of which strictly prefer $S$ to where they are in $\pi$. Intuitively, if the state of affairs were $\pi$, then the members of $S$ would find each other and decide to defect together because $\pi$ does not offer them enough utility. In this case, we say that $S$ \textit{blocks}~$\pi$. A related concept is the \textit{strict core}. The partition $\pi$ is strict-core-stable if there is no set $S$ of agents all of which \textit{weakly} prefer $S$ to $\pi$, but with at least one agent $i$ in $S$ having a strict preference in favour of $S$. Intuitively, $i$ can offer some of his profit of the deviation to the other players, inducing $S$ to deviate from $\pi$; and even if utility is entirely non-transferable, the players who are indifferent between $\pi$ and $S$ will be easily convinced to join the deviation.

Computationally speaking, hedonic games are large objects: every player is described by a preference ordering over exponentially many sets. However, there are attractive concise representations of such preferences, many of which have been studied in the computer science literature \citep{Ballester2004,Cechlarova2004,Elkind2009,Sung2010,Aziz2014a}. Given a concise representation (which usually is not universally expressive), it makes sense to pose the computational problem of finding a stable outcome given a hedonic game. Since many games do not admit stable outcomes, we can conveniently consider the decision problem whether one exists. This problem turns out to be NP-hard for most cases (see Peters and Elkind \citet{Peters2015} for a systematic study of representations inducing NP-hard \textsc{core-existence} problems, and Woeginger \citet{WoegingerSurvey} for a survey). In particular, the case of \textit{additively separable hedonic games} is an example of a class of succinctly represented hedonic games where it is hard to distinguish games admitting stable outcomes from those which do not \citep{Sung2010,ABS11c}. In this model, agents assign numeric values $v_i(j)$ to other players, and the utility of a coalition is the sum $\sum_{j\in S} v_i(j)$ of the values of the players in it.

In general, \textsc{core-existence}, the problem of deciding whether a given hedonic game admits a core-stable partition, is contained in the complexity class $\Sigma_2^p$. This is because the question under consideration is characterised by a single alternation of quantifiers: ``does there \textit{exist} a partition $\pi$ such that \textit{for all} coalitions $S$, $S$ does not block $\pi$?''. Here, we have assumed to be able to efficiently decide whether a given coalition $S$ blocks, which is trivially the case for all commonly considered representations. Alternatively, since $\Sigma_2^p = \text{NP}^{\text{NP}}$, we can see containment of \textsc{core-existence} through the following non-deterministic algorithm: guess a partition $\pi$, and use the NP-oracle to check whether $\pi$ is core-stable. 

In the case of additively separable games, it is coNP-complete to verify that a given partition $\pi$ is core-stable \citep{Dimitrov2006}. Thus, it is unlikely that the \textsc{core-existence} problem is contained in NP. Indeed, Woeginger \citet{Woeginger2013} proved that the problem is $\Sigma_2^p$-complete. The problem thus encapsulates the full hardness of the second level of the polynomial hierarchy, making it much harder to decide than NP-complete problems. (Woeginger likens NP-complete problems to ``rotten eggs'' while $\Sigma_2^p$-complete problems are at the level of ``radioactive thallium''.) 

Recent decades have shown impressive advances in general-purpose tools to handle NP-complete problems, such as through SAT and ILP solvers. Thus, for solution concepts such as Nash stability, whose \textsc{existence} problem is NP-complete, we should expect tractability in many practically relevant cases. On the other hand, because of $\Sigma_2^p$-hardness, finding a core-stable partition will likely require exponentially many calls to solvers. This suggests that the core will remain computationally elusive for some time to come.

In this paper, we extend Woeginger's result to the \emph{strict core}. In addition, our reduction---which works for both the core and the strict core---produces additive hedonic games that are \emph{symmetric} and \emph{sparse}, showing that imposing these additional restrictions do not lead to a drop in complexity. Before we turn to additive games, however, we first consider hedonic games with \emph{dichotomous preferences}, or \emph{Boolean hedonic games} \cite{Booleanhedonicgames2014}. In these games, players only distinguish between \emph{approved} and \emph{non-approved} coalitions. Preferences in these games can be specified by giving each agent a goal formula of propositional logic; the approved coalitions are those which satisfy the goal. We show that deciding the existence of a strict-core-stable outcome in a Boolean hedonic game is $\Sigma_2^p$-complete. For the framework of Aziz et al.\ \cite{Booleanhedonicgames2014}, this is interesting because it suggests that no polynomial-sized formula in their logic will be able to characterize strict-core-stable outcomes. Our hardness result for \emph{additive} games is then obtained by implementing the reduction for the dichotomous case using additive valuations.



\section{Preliminaries}

Given a finite set of agents $N$, a \textit{hedonic game} is a pair 
$G = \langle N, (\pref_i)_{i\in N}\rangle$, where each agent $i\in N$ possesses 
a complete and transitive preference relation $\pref_i$ over $\mathcal N_i = \{ S \subseteq N : i \in S \}$, the set of coalitions containing $i$. If $S \pref_i T$, we say that $i$ \textit{weakly prefers} $S$ to $T$. If $S \pref_i T$ but $T\not\pref_i S$, we say that the preference is \textit{strict} and write $S \succ_i T$. An agent is \textit{indifferent} between $S$ and $T$ whenever both $S\pref_i T$ and $T\pref_i S$.

An {\em outcome} of a hedonic game is a partition $\pi$ of $N$ into disjoint coalitions. We write $\pi(i)$ for the coalition of $\pi$ that contains $i$. If $\pi(i) \pref_i \{i\}$ for each $i\in N$, then $\pi$ is \textit{individually rational}. We say that a non-empty coalition $S\subseteq N$ \textit{blocks} $\pi$ if $S \succ_i \pi(i)$ for all $i\in S$. Thus, all members of a blocking coalition strictly prefer that coalition to where they are in $\pi$. The partition $\pi$ is \textit{core-stable} if there is no blocking coalition. A non-empty coalition $S\subseteq N$ \textit{weakly blocks} $\pi$ if $S \pref_i \pi(i)$ for all $i\in S$, and $S \succ_i \pi(i)$ for some $i\in S$. The partition $\pi$ is \textit{strict-core-stable} if there is no weakly blocking coalition. Clearly, if $\pi$ is strict-core-stable, then it is also core-stable. There are many other stability concepts for hedonic games that we do not consider here, see Aziz and Savani \citet{azizsavani} for a survey.

A hedonic game has \textit{dichotomous preferences}, and is called a \textit{Boolean hedonic game}, if for each agent $i\in N$ the coalitions $\mathcal N_i = \{ S \subseteq N : i \in S \}$ can be partitioned into \textit{approved} coalitions $\mathcal N_i^+$ and non-approved coalitions $\mathcal N_i^-$ such that $i$ strictly prefers approved coalitions to non-approved coalitions, but is indifferent within the two groups: so $S \succ_i T$ iff $S\in \mathcal N_i^+$ and $T\in \mathcal N_i^-$. We can specify a Boolean hedonic game by assigning to each agent $i\in N$ a formula $\phi_i$ ($i$'s \textit{goal}) of propositional logic with the propositional atoms given by the agent set $N$. A coalition $S\ni i$ is then approved by $i$ if and only if $S\models \phi_i$, that is the formula $\phi_i$ is satisfied by the truth assignment that sets variable $j\in N$ true iff $j\in S$. For example, if $i$'s goal formula $\phi_i$ is $(j \lor k) \land \lnot \ell$, then $i$ approves coalitions containing agent $j$ or $k$ as long as they do not contain agent $\ell$.

A hedonic game is \textit{additively separable}, and is called an \textit{additive hedonic game}, if there are valuation functions $v_i : N \to \mathbb Z$ for each agent $i\in N$ such that $S\pref_i T$ if and only if $\sum_{j\in S} v_i(j) \ge \sum_{j\in T} v_i(j)$. An additive hedonic game is \textit{symmetric} if $v_i(j) = v_j(i)$ for all $i,j\in N$.

The complexity class $\Sigma_2^p$, the \textit{second level of the polynomial hierarchy}, is  $\text{NP}\raisebox{-0.3pt}{$^{\text{NP}}$}$, the class of problems solvable by a polynomial-time non-deterministic Turing machine when given an NP-oracle. It can also be seen as the class of problems polynomial-time reducible to the language TRUE $\exists\forall$-QBF, which consists of true quantified Boolean formulas with only 1 alternation of quantifiers.

\section{Related Work}
Boolean hedonic games were introduced by Aziz et al.\ \citet{Booleanhedonicgames2014} who study them from a mainly logical point of view. (Notice that they use a different choice of propositional atoms---\textit{pairs} of agents---to allow future generalisations to games played on general coalition structures. Our choice is more natural for the hedonic setting.) In particular, Aziz et al.\ \citet{Booleanhedonicgames2014} show that every Boolean hedonic game admits a core-stable partition. Thus, only the complexity of the existence of the \textit{strict} core needs to be settled. Peters \citet{peters2016dichotomous} shows that \emph{finding} a core-stable partition, while it is guaranteed to exist, is FNP-hard.

Elkind and Wooldridge \citet{Elkind2009} introduce a representation formalism for hedonic games called \textit{hedonic coalition nets} (HC-nets), which can be seen as a powerful mixture of the additive and Boolean representations introduced above. Here, agents provide several goals $\phi_i$ weighted by real numbers, and an agent obtains as utility the sum of the weights of the formulas that are satisfied by the coalition. Thus, additively separable games are given by HC-nets in which every formula is just a single positive literal, and Boolean hedonic games are given by HC-nets in which every agent has only a single formula. It follows from a general result of Malizia et al.\ \citet{malizia2007infeasibility} that core-existence is $\Sigma_2^p$-complete to decide for games given by HC-nets (see also Elkind and Wooldridge \cite{Elkind2009}). We strengthen this to also apply to the strict core, and to hold even if either every agent only has a single formula, or if every formula is given by a single literal.

This paper follows the work of Woeginger \citet{Woeginger2013} who proves that deciding core-existence is $\Sigma_2^p$-complete for additive hedonic games. His reduction (from the same problem that we reduce from) does not work for the strict core, and in his survey \citep{WoegingerSurvey} he poses the problem to establish $\Sigma_2^p$-hardness for this solution concept. Doing this is the main contribution of this paper. We also strengthen Woeginger's \citep{Woeginger2013}
result for the core to hold even for symmetric valuations, and even if $v_i(j)$ is non-zero for at most 10 agents $j$, so that the game is ``sparse''. This closes off two avenues for potential avoidance of $\Sigma_2^p$-hardness.

Since a preprint of this paper appeared on arXiv, some additional $\Sigma_2^p$-hardness results for hedonic games have been obtained. Ohta et al.\ \citet{ohta2017friends} show that in hedonic games based on aversion to enemies (introduced by Dimitrov et al.\ \citet{Dimitrov2006}), if one allows `neutral' players, then deciding the existence of the (strict) core is $\Sigma_2^p$-hard. The games studied by Ohta et al.\ \citet{ohta2017friends} are in fact additively separable; thus, they imply $\Sigma_2^p$-hardness of the (strict) core in additive hedonic games. Hardness holds even if we restrict players' valuations to take at most three values, so that $v_i(j) \in \{ -n, 0, 1 \}$ for all $i,j\in N$. However, in contrast to our reduction, the games produced in their reduction are not symmetric and not sparse. Ohta et al.\ \citet{ohta2017friends} also show $\Sigma_2^p$-hardness for the strict core for games based on friend appreciation in the presence of neutral players. Aziz et al.\ \citet{aziz2017fractional} consider \emph{fractional hedonic games}, where players care about the \emph{average} value of their coalition partners, rather than the sum. They show that deciding the existence of the core is $\Sigma_2^p$-hard, even if valuations are symmetric and simple, so that $v_i(j) \in \{0,1\}$ for all $i,j\in N$. The reductions in both of these recent papers is from the complement of the \textsc{minmax clique} problem \citep{ko1995complexity}, which seems to be well-suited as a starting point for reductions for hedonic games. In this paper, like in Woeginger \citet{Woeginger2013}, we instead use a problem based on quantified Boolean formulas.

Peters \citet{peters2016graphical} introduces \emph{graphical hedonic games}, which are hedonic games equipped with an underlying graph on the agent set. This is a direct analogue of non-cooperative graphical games. In this language, our result for additive hedonic games implies that deciding the existence of the (strict) core is $\Sigma_2^p$-hard even for graphical hedonic games of bounded degree---that is, games whose underlying graph has a max-degree of at most $10$.

\section{A Useful Restricted Hard Problem}
Stockmeyer \citet{stockmeyer1976polynomial} proved that the following basic problem is $\Sigma_2^p$-complete: 

\noindent
\rule{\textwidth}{0.8pt}
\textbf{TRUE $\boldsymbol{\exists\forall}$-3DNF} \\[-6pt]
\rule{\textwidth}{0.5pt}
\textbf{Instance:} A quantified Boolean formula of form
\[
\setlength{\abovedisplayskip}{2pt}
\setlength{\belowdisplayskip}{3pt}
\exists x_1,\dots,x_m\:\: \forall y_1,\dots,y_n\:\: \phi(x_1,\dots,x_m,y_1,\dots,y_n),\] 
where $\phi$ is in disjunctive normal form with each disjunct containing 2 or 3 literals. \\[3pt]
\textbf{Question:} Is the formula true? \\[-6pt]
\rule{\textwidth}{0.8pt}

\noindent
Here we show that the problem remains $\Sigma_2^p$-complete even if we place restrictions on the number of occurrences of the variables, like is standard practice when proving NP-completeness. 

\noindent
\rule{\textwidth}{0.8pt}
\textbf{RESTRICTED TRUE $\boldsymbol{\exists\forall}$-3DNF} \\[-6pt]
\rule{\textwidth}{0.5pt}
\textbf{Instance:} A quantified Boolean formula of form
\[
\setlength{\abovedisplayskip}{2pt}
\setlength{\belowdisplayskip}{3pt}
\exists x_1,\dots,x_m \forall y_1,\dots,y_n\:\: \phi(x_1,\dots,x_m,y_1,\dots,y_n),\] 
where $\phi$ is in disjunctive normal form with
\begin{itemize}
	\itemsep0em
	\item each disjunct containing 2 or 3 literals,
	\item each $x$-variable occurring exactly once positive and once negative
	\item each $y$-variable occurring exactly three times, and at least once positively and at least once negatively.
\end{itemize}\vspace{-2pt}
\textbf{Question:} Is the formula true? \\[-6pt]
\rule{\textwidth}{0.8pt}

\noindent
We may further insist that every disjunct contain at most 2 $x$-literals, because a disjunct containing only $x$-literals makes the formula trivially true.

\begin{proposition}
	The problem \textup{RESTRICTED TRUE ${\exists\forall}$-3DNF} is $\Sigma_2^p$-complete.
\end{proposition}
\begin{proof}
	Membership in $\Sigma_2^p$ is clear.
	
	Let us note that a true \textsc{\scalebox{0.83}{$\exists\hspace{1pt}\forall$}-\scalebox{0.83}{3}dnf}-formula is the same thing as a \textit{false} \scalebox{0.83}{$\forall\exists$}-\textsc{\scalebox{0.83}{3}cnf}-formula; we will use this latter view since CNF formulas are more familiar. Thus, we may reduce from the unrestricted problem \textsc{false-\scalebox{0.83}{$\forall\exists$}-\scalebox{0.83}{3}cnf}. We will in polynomial time transform a given such formula \[\forall x_1,\dots,x_m \exists y_1,\dots,y_n\:\: \phi(x_1,\dots,x_m,y_1,\dots,y_n)\] into a formula of equal truth value which is restricted as above, establishing hardness of the restricted variant.
	
	First, by unit propagation, we may assume that no clause contains only a single literal. Next, for each $x$-variable $x_i$, relabel all its occurrences as $y_i^1,\dots,y_i^{n_i}$ where $n_i$ is the number of occurrences of $x_i$, and the $y_i^r$ are new variables. Existentially quantify over these new variables, keeping $x_i$ universally quantified. Add clauses $(x_i\to y_i^1) \land (y_i^1 \to y_i^2) \land\cdots\land(y_i^{n_i} \to x_i)$ to force all copies to have the same truth value. Similarly, for each old $y$-variable, we relabel all its occurrences (existentially quantifying) and add a `wheel of implications' for them as well (here we may discard the old $y$-variable in the process). The resulting formula satisfies the restrictions and has the same truth value as the original formula.
\end{proof}

By using the techniques of Berman et al.\ \citet{Berman2003}, we can similarly prove that the problem remains $\Sigma_2^p$-complete if disjuncts are required to contain exactly 3 distinct literals, each $x$-literal occurs exactly once, and each $y$-literal occurs exactly twice. One can also show that the problem remains $\Sigma_2^p$-complete if every clause contains at most one $x$-literal \cite[Theorem 10]{ko1995complexity}. For the reductions in this paper, we do not need these other restrictions. 

\section{Strict Core for Boolean Hedonic Games}

Our first hardness result concerns Boolean hedonic games, as introduced by Aziz et al.\ \citet{Booleanhedonicgames2014}. While for this type of hedonic game, the core is always guaranteed to exist, the \emph{strict} core is more difficult to handle.

\begin{theorem}
	The problem ``does a given Boolean hedonic game admit a strict-core-stable partition?'' is $\Sigma_2^p$-complete.
\end{theorem}
\begin{proof}
	Membership in $\Sigma_2^p$ is clear, since we are asking: does there \textit{exist} a partition such that \textit{for all} coalitions $S$, $S$ does not block?
	
	For hardness, we reduce from RESTRICTED TRUE ${\exists\forall}$-3DNF. Let $\varphi = \exists x\forall y \phi$ be an instance of this problem, where $x=(x_i)$ and $y=(y_j)$ denote vectors of variables. We rewrite $\varphi$ as $\exists x (\lnot \exists y \lnot \phi)$. Note that $\lnot \phi$ is a 3CNF formula, and when below we talk about clauses, we are always referring to clauses of $\lnot\phi$. In the hedonic game which we construct below, a strict-core-stable partition corresponds to an assignment to the $x$-variables. If there is a $y$-assignment satisfying $\lnot \phi$  (so that $\varphi$ is false), this will form a weakly blocking coalition, and conversely, such a coalition induces a satisfying assignment. If the formula is true, such a blocking coalition cannot exist.
	
	For our construction, we take the following agents:
	\begin{itemize}
		\item For each $x_i$, four agents $x_i$, $\overline x_i$, $t_i$, $f_i$.
		\item For each $y_j$, two agents $y_j$, $\overline y_j$.
		\item For each clause $c_k$ in $\lnot\phi$, one agent $c_k$.
		\item A single player $\varphi$ representing the formula.
	\end{itemize}
	We now specify agents' goals. For a clause $c_k$ of $\lnot\phi$, we let $\ell_1^k,\ell_2^k,\ell_3^k$ denote the agents corresponding to the literals occurring in it. For example, if clause $c_k$ is $(x_1 \lor \lnot y_2 \lor y_3)$, then $\ell_1^k$ refers to agent $x_1$, $\ell_2^k$ refers to $\overline y_2$, and $\ell_3^k$ refers to $y_3$. If a clause only contains 2 literals, just let $\ell_2^k = \ell_3^k$.
	\begin{itemize}
		\item $x_i: f_i \land \lnot t_i \land \overline x_i$
		\item $\overline x_i: t_i \land \lnot f_i \land x_i$
		\item $t_i: \lnot\varphi$
		\item $f_i: \lnot\varphi$
		\item $y_j: \lnot \overline y_j$
		\item $\overline y_j: \lnot y_j$
		\item $c_k: \lnot\varphi \lor ((\ell_1^k \lor \ell_2^k \lor \ell_3^k)\land c_{k+1})$, or \\
		$c_k: \lnot\varphi \lor (\ell_1^k \lor \ell_2^k \lor \ell_3^k)$ if $c_{k+1}$ does not exist
		\item $\varphi: c_1\land (x_1\lor \overline x_1)$
	\end{itemize}
	This hedonic game has a strict-core-stable outcome if and only if $\varphi$ is true.
	
	$\impliedby$: Suppose $\varphi$ is true. Take an assignment $\mathcal A$ to the $x$-variables certifying truth of $\varphi$. Then take the partition $\pi$ with coalitions $\{t_i, x_i, \overline x_i\}$ for true $x_i$, with coalitions $\{f_i, x_i, \overline x_i\}$ for false $x_i$, and singleton coalitions for all other players. We show that $\pi$ is strict-core-stable. 
	
	Most agents' goals are satisfied in $\pi$, except for true $x$-literals and the player $\varphi$. If $\pi$ is not stable, then there is a weakly blocking coalition $S$ including a player whose goal is satisfied in $S$ but not in $\pi$; also, no other player in $S$ can be worse off in $S$ than in $\pi$. Now the profiting player cannot be a true $x$-literal, for if this player were to gain then its complementary literal must be part of $S$ and this literal would lose, because complementary literals have incompatible goals. Hence any weakly blocking coalition $S$ must include the $\varphi$-player and must satisfy it. Looking at $\varphi$'s goal, this means that $c_1\in S$. Indeed, by induction, $c_k\in S$ for all $c_k$, since $c_k\in S$ cannot be worse off in $S$, and thus $c_k$'s goal must be satisfied, which requires $c_{k+1}\in S$. 
	
	Since every $c_k$ is satisfied in $\pi$, the goal of every $c_k$ is also satisfied in $S$. This means that the literals present in $S$ must satisfy each clause of $\lnot\phi$. Thus, the literals present in $S$ satisfy $\lnot \phi$. This contradicts truth of $\varphi$ under the assignment $\mathcal A$ once we can show that $S$ does not contain complementary $y$-literals, and only contains $x$-literals that are true in $\mathcal A$. But both of these requirements are easy to see: since all $y$-literals have their goal satisfied in $\pi$, they must have their goal satisfied in $S$, which means their complementary literal is not part of $S$. Also, a false $x$-literal is happy in $\pi$ but would be unhappy in $S$ since $t_i,f_i\not\in S$ (because they hate $\varphi$), and thus false $x$-literals are not part of the weakly blocking $S$. Hence $\pi$ is strict-core-stable.
	
	$\implies$: Suppose the game has a strict-core-stable outcome $\pi$. We show that $\varphi$ is true. 
	First we will find an appropriate assignment to the $x$-variables. Fix some variable $x_i$. Since the goals of $x_i$ and $\overline x_i$ are incompatible, at most one of them is happy in $\pi$. If both are unhappy, then $\{t_i,x_i,\overline x_i\}$ weakly blocks. Hence for each $x_i$, exactly one of $x_i$ and $\overline x_i$ has their goal satisfied. Define the assignment that sets that literal true which is \textit{not} satisfied.
	
	Soon we will need to know that the $\varphi$ player does not have its goal satisfied in $\pi$. For a contradiction suppose it does. Then $x_1\in\pi(\varphi)$ or $\overline x_1 \in \pi(\varphi)$. Since $x_1$ or $\overline x_1$ has their goal satisfied in $\pi$, both of them are together with either $t_1$ or $f_1$. So either $t_1\in\pi(\varphi)$ or $f_1\in\pi(\varphi)$; but then $\{t_1\}$ or $\{f_1\}$ blocks $\pi$, a contradiction.
	
	Now take an arbitrary assignment to the $y$-variables, and suppose for a contradiction that under these assignments to the $x$- and $y$-variables the formula $\phi$ becomes false, so that $\lnot \phi$ becomes true so every clause is true. Let $S$ be the coalition consisting of player $\varphi$, all clauses $c_k$, all true $x$-literals, and all true $y$-literals. In $S$, every player except for the true $x$-literals is satisfied, so no player is worse off. However $\varphi$ did not have its goal satisfied in $\pi$, so is strictly better off in $S$, and hence $S$ weakly blocks $\pi$, a contradiction. Thus, $\varphi$ must be true.
\end{proof}
Hardness holds even if every agent mentions at most 5 other agents in their goal. By rewriting the formulas, we can see that hardness also holds even if goals are given in 3-DNF or 4-CNF.

\section{Core and Strict Core for Additive Hedonic Games}

The structure of our reduction for Boolean hedonic games can be adapted to work in the additive case. The resulting reduction is necessarily less straightforward, because we have to simulate the clausal structure using additive valuations. On the other hand, the resulting reduction works for \emph{both} the core and the strict core. Further, this is the first hardness reduction for additive hedonic games that applies even to ``sparse'' games, where players assign non-zero valuations to only at most a fixed number of other players.

\begin{theorem}
	\label{thm:additive}
	The problem ``does a given additive hedonic game admit a strict-core-stable partition?'' is $\Sigma_2^p$-complete. The same question for the core is also $\Sigma_2^p$-complete, and both problems remain hard even for symmetric utilities that only assign non-zero values to at most 10 other players.
\end{theorem}
\begin{proof}
	Membership in $\Sigma_2^p$ is clear, since we are asking: does there \textit{exist} a partition such that \textit{for all} coalitions $S$, $S$ does not block?
	
	For hardness, we reduce from RESTRICTED TRUE ${\exists\forall}$-3DNF. Let $\varphi = \exists x\forall y \phi$ be an instance of this problem, where $x=(x_i)$ and $y=(y_j)$ denote vectors of variables. We rewrite $\varphi$ as $\exists x (\lnot \exists y \lnot \phi)$. Note that $\lnot \phi$ is a 3CNF formula, and when below we talk about clauses, we are always referring to clauses of $\lnot\phi$. In the hedonic game which we construct below, a (strict-)core-stable partition corresponds to an assignment to the $x$-variables. If there is a $y$-assignment satisfying $\lnot \phi$, this will form a blocking coalition, and conversely, such a coalition induces a satisfying assignment. If the formula $\varphi$ is true under the assignment to the $x$-variables, such a blocking coalition cannot exist.
	
	Before we start, let us add a new $x$-variable (call it~$x_*$) to the formula and add the clause $c_1 := (x_*\lor\overline x_*)$ to $\lnot\phi$. This preserves the truth value of the formula, and preserves the restrictions of the input problem. This will help in the proof of implication (iii) $\Rightarrow$ (i) later.
	
	We take the following agents:
	\begin{itemize}
		\item For each $x_i$, four agents $x_i$, $\overline x_i$, $t_i$, $f_i$.
		\item For each $y_j$, two agents $y_j$, $\overline y_j$.
		\item For each clause $c_k$ in $\lnot\phi$, one agent $c_k$.
		\item A helper player $c_k'$ for each $c_k$, and helper players $t_i'$ and $f_i'$ for each $t_i$ and $f_i$.
	\end{itemize}
	We say that a player $p$ who has a helper player $p'$ is \emph{supported}. The purpose of the helper players will be to guarantee that supported players obtain utility at least 30 in stable outcomes.

	In Table~\ref{table:valuations}, we specify agents' symmetric utilities, see also Figures~\ref{fig:clause} and~\ref{fig:variable}. All utilities not specified in the table are 0; the figures also omit $-\infty$ valuations for clarity. As notation, for a literal $\ell$, we let $c(\ell)$ denote the clauses that contain $\ell$; for $x$-literals there is just one such clause, but there are up to two for $y$-literals. We also let $\ell(c)$ denote the literals occurring in $c$. We take arithmetic in subscripts of clause names (as in `$c_{k+1}$') to be modulo the number of clauses in $\lnot\phi$.
	\begin{table}[t]
	\centering
	\begin{tabular}{l l r}
		\toprule
		$a$ & $b$ & $v_a(b)$  \\
		\midrule
		$x_i$ & $\overline x_i$ & $-10$ \\
		& $f_i$ & 20 \\
		& $t_i$ & 14 \\
		& $c(x_i)$ & 5 \\
		& $f_i'$ & $-\infty$ \\
		& $t_i'$ & $-\infty$  \\
		& $c'(x_i)$ & $-\infty$ \\
		\midrule
		$t_i$ & $x_i$ & 14 \\
		& $\overline x_i$ & 20  \\
		& $t_i'$ & 30 \\
		& $f_i$ & $-\infty$ \\
		& $c(x_i)$ & $-\infty$ \\
		\bottomrule \\
	\end{tabular}\qquad%
\begin{tabular}{l l r}
	\toprule
	$a$ & $b$ & $v_a(b)$  \\
	\midrule
	$\overline x_i$ & $x_i$ & $-10$ \\
	 & $f_i$ & 14 \\
	& $t_i$ & 20 \\
	& $c(\overline x_i)$ & 5 \\
	& $f_i'$ & $-\infty$ \\
	& $t_i'$ & $-\infty$ \\
	& $c'(\overline x_i)$ & $-\infty$ \\
	\midrule
	$f_i$ & $x_i$ & 20 \\
	& $\overline x_i$ & 14 \\
	& $f_i'$ & 30 \\
	& $t_i$ & $-\infty$ \\
	& $c(\overline x_i)$ & $-\infty$ \\
	\bottomrule \\
\end{tabular}\qquad%
\begin{tabular}{l l r}
	\toprule
	$a$ & $b$ & $v_a(b)$  \\
	\midrule
	$y_j$ & $\overline y_j$ & $-\infty$ \\
	& $c(y_j)$ & $5$ \\
	& $c'(y_j)$ & $-\infty$ \\
	\midrule
	$c_k$ & $c_{k-1}$ & 13 \\
	& $c_{k+1}$ & 13 \\
	& $\ell$(c) & 5 \\
	& $c_k'$ & 30 \\
	& $t_i/f_i$ & $-\infty$ \\
	& $c_{k-1}'$ & $-\infty$  \\
	& $c_{k+1}'$ & $-\infty$ \\
	\bottomrule \\
	&&\\
	&&\\
\end{tabular}\qquad%
\begin{tabular}{l l r}
	\toprule
	$a$ & $b$ & $v_a(b)$  \\
	\midrule
	$\overline y_j$ & $y_j$ & $-\infty$ \\
	& $c(\overline y_j)$ & $5$ \\
	& $c'(\overline y_j)$ & $-\infty$ \\
	\midrule
	$c_k'$ & $c_k$ & 30 \\
	& $\ell(c_k)$ & $-\infty$ \\
	& $c_{k-1}$ & $-\infty$ \\
	& $c_{k+1}$ & $-\infty$ \\
	\midrule
	$t_i'$ & $t_i$ & 30 \\
	& $x_i$, $\overline x_i$ & $-\infty$ \\
	\midrule
	$f_i'$ & $f_i$ & 30 \\
	& $x_i$, $\overline x_i$ & $-\infty$ \\
	\bottomrule \\
\end{tabular}
	\caption{The agent valuations $v_a(b)$. All values not specified are 0. The value ``$-\infty$'' denotes any sufficiently large negative number; $-100$ will do. Notice that the valuations are symmetric ($v_a(b) = v_b(a)$) and every agent specifies at most 10 non-zero values, noting that we can ensure that no clause contains more than 2 $x$-literals.}
	\label{table:valuations}
	\end{table}
	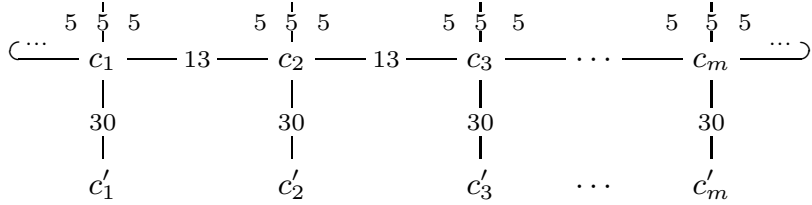
\begin{figure}[t]
		\centering
		\scalebox{1.3}{
		\xymatrix@C-0.55pc{
			&&&&&&&& \\
			\ar@{*{\UseTips\dir^{(}}-}[r]^{\cdots\quad\:}
			& c_1 \ar@{-}[rr]|{\lab{13}} 
			\ar@{}[ul]^(.07){}="a"^(.6){}="b" \ar@{-}|{\labs{5}} "a";"b"
			\ar@{}[u]^(.07){}="a"^(.6){}="b" \ar@{-}|{\labs{5}} "a";"b"
			\ar@{}[ur]^(.07){}="a"^(.6){}="b" \ar@{-}|{\labs{5}} "a";"b"
			 && c_2 \ar@{-}[rr]|{\lab{13}} 
			 \ar@{}[ul]^(.07){}="a"^(.6){}="b" \ar@{-}|{\labs{5}} "a";"b"
			 \ar@{}[u]^(.07){}="a"^(.6){}="b" \ar@{-}|{\labs{5}} "a";"b"
			 \ar@{}[ur]^(.07){}="a"^(.6){}="b" \ar@{-}|{\labs{5}} "a";"b"
			 && c_3 \ar@{-}[r]
			 \ar@{}[ul]^(.07){}="a"^(.6){}="b" \ar@{-}|{\labs{5}} "a";"b"
			 \ar@{}[u]^(.07){}="a"^(.6){}="b" \ar@{-}|{\labs{5}} "a";"b"
			 \ar@{}[ur]^(.07){}="a"^(.6){}="b" \ar@{-}|{\labs{5}} "a";"b"
			 & \cdots \ar@{-}[r] & 
			 c_m
			 \ar@{}[ul]^(.07){}="a"^(.6){}="b" \ar@{-}|{\labs{5}} "a";"b"
			 \ar@{}[u]^(.07){}="a"^(.6){}="b" \ar@{-}|{\labs{5}} "a";"b"
			 \ar@{}[ur]^(.07){}="a"^(.6){}="b" \ar@{-}|{\labs{5}} "a";"b" 
			 \ar@{-*{\UseTips\dir^{)}}}[r]^{\quad\:\cdots} & \\
			&c_1' \ar@{-}[u]|{\lab{30}} && c_2' \ar@{-}[u]|{\lab{30}} && c_3' \ar@{-}[u]|{\lab{30}} & \cdots &  c_m' \ar@{-}[u]|{\lab{30}}
			}
		}
		\caption{The clause gadgets. Clauses are arranged in a cycle. We will consider a partition $\pi$ where each clause agent $c_k$ is in a pair with its helper $c_k'$. If the formula is false, though, all the clauses join forces and can deviate together with a falsifying selection of $y$-literals (and of true $x$-literals).}
		\label{fig:clause}
	\end{figure}
	
	\begin{figure}[t]
		\centering
		\scalebox{1.1}{\xymatrix@R-0.3pc@L=10pt{
			& & t_i' \ar@{-}[d]|{\lab{30}} & & \\
			& & t_i \ar@{-}[dl]|{\lab{14}} \ar@{-}[dr]|{\lab{20}} & & \\
			c(x_i)
			& \ar@{-}[l]|{\lab{5}} x_i \ar@{-}[rr]|{\:\:-10\:\:} 
			& & \overline x_i \ar@{-}[r]|{\lab{5}} & c(\overline x_i)  \\
			& & f_i \ar@{-}[ur]|{\lab{14}} \ar@{-}[ul]|{\lab{20}} & & \\
			& & f_i' \ar@{-}[u]|{\lab{30}} & &
		}}
		\caption{The $x$-variable gadget. If the variable is set true, the upper triangle forms a coalition. If it is false, the lower triangle forms a coalition. Note that in this configuration, the true (but not the false) literal is willing to deviate with the connected clause agent.}
		\label{fig:variable}
	\end{figure}
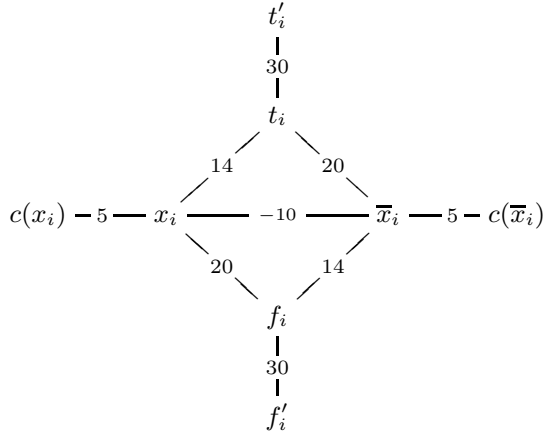
	
	\noindent
	We show that the following are equivalent:
	\begin{enumerate}[(i)]
		\item The input formula is true.
		\item The game admits a strict-core-stable partition.
		\item The game admits a core-stable partition.
	\end{enumerate}
	
	Call a coalition $S$ \textit{feasible} if it is individually rational (and in particular does not contain players who evaluate each other as $-\infty$), and if each `supported' player (those that have a helper player: $c_k/t_i/f_i$) obtains utility $\ge 30$ in $S$. Call a coalition \textit{infeasible} otherwise. Observe that all coalitions in a core-stable partition must be feasible (as otherwise it is not individually rational or a coalition like $\{c_k, c_k'\}$ blocks). Observe further that if a partition $\pi$ is weakly blocked by some coalition $S$, then it is weakly blocked by a feasible coalition. This is because if $S$ is not individually rational for player $i$, then $\pi$ is also weakly blocked by the feasible coalition $\{i\}$ or $\{i, i'\}$, and if $S$ gives less than utility 30 to a supported player $i$, then the feasible coalition $\{i, i'\}$ also weakly blocks $\pi$. These properties of feasible coalitions formalise the intuitive purpose of the helper players: they guarantee a minimum payoff to the supported player in any stable outcome. \\
	
	\noindent
	\textbf{Lemma.} In a feasible coalition $S\ni c_k$, either \\
	(a) $c_k' \in S$ but $c_{k-1}, c_{k+1}, \ell(c_k)\not\in S$, or \\ 
	(b) $S$ contains \textit{all} clause players simultaneously, and for each clause, some literal occurring in it is part of $S$. \vspace{3pt}
	
	\noindent 
	\textit{Proof.} Note that $c_k'$ hates all other players to which $c_k$ assigns positive utility (namely $c_{k-1}, c_{k+1}, \ell(c_k)$), so that if $c_k'\in S$ then those players cannot be in $S$. So suppose $c_k'\not\in S$. We will prove by induction that we are in case (b). By feasibility of $S$, $c_k$ obtains utility at least 30. This is only possible if both $c_{k-1}\in S$ and $c_{k+1}\in S$ (subscripts mod the number of clauses), since the literals connected to $c_k$ only give total utility $3 \cdot 5 = 15$. Thus, $c_{k+1} \in S$. Assume now that $c_k, c_{k+1},\dots,c_{s-1},c_s\in S$ for some $s$. As $c_{s-1}$ and $c_s'$ hate each other, we have $c_s'\not\in S$. But feasibility for $c_s$ then implies that $c_{s+1}\in S$. By induction, all clause players are in $S$. Finally, if some clause player did not have any of its literals in $S$, then it would only obtain utility $26$ in $S$ contradicting feasibility. \qed \vspace{8pt}
	
	\noindent
	Let us note that in case (b), some $x$-literal must be part of $S$; namely at least one $x_*$-literal by feasibility for $c_1$.\\
	
	\noindent
	(i) $\Rightarrow$ (ii): Suppose the input formula is true. Take an assignment $\mathcal A$ of the $x$-variables certifying truth of the formula. Then take the partition $\pi$ with coalitions $\{t_i, x_i, \overline x_i\}, \{t_i'\}, \{f_i,f_i'\}$ for true $x_i$, with coalitions $\{f_i, x_i, \overline x_i\}, \{f_i'\},\{t_i,t_i'\}$ for false $x_i$, with coalitions $\{c_k,c_k'\}$ for each $c_k$, and singleton coalitions $\{y_j\},\{\overline y_j\}$ for $y$-literals. We show that $\pi$ is strict-core-stable. 
	
	Suppose not, and there is a weakly blocking coalition~$S$ which we may assume to be feasible. In $\pi$, the following players are already in a best coalition among feasible ones: false $x$-literals (with utility $10$), $t_i$ and $f_i$ players together with literals (with utility $34$), $c_k'$ (with utility $30$), and paired $t_i', f_i'$ (with utility $30$). Since $S$ is weakly blocking, it contains a player $p$ who \textit{strictly} prefers $S$ to~$\pi$. This $p$ cannot be on the preceding list.
	
	Now $p$ also cannot be a $t_i$ or $f_i$ player together with its helper for it cannot offer its associated literals a feasible coalition as good as they have in $\pi$ (namely, the false literal involved would move from utility $10$ to $4$). Also, $p$ cannot be a singleton $t_i'$ or $f_i'$ as $t_i$ or $f_i$ (currently with utility $34$) will be less happy in feasible coalitions together with $t_i'$ or $f_i'$ (which only offers $30$). Hence, $p$ is either a true $x$-literal, a $y$-literal, or a $c_k$ player. In the former two cases, in order that $p$ strictly gains utility, $S$ must also include a clause player associated with the literal $p$. So in either case $S$ includes a clause player, and by the lemma, since we cannot be in case (a), $S$ must contain \textit{all} clause players, and enough literals to satisfy every clause. This means that the literals present in $S$ satisfy $\lnot \phi$. This contradicts truth of $\varphi$ under assignment $\mathcal A$ once we can show that the literals in $S$ form a valid assignment, i.e., $S$ does not contain complementary $y$-literals, and only contains true $x$-literals. But both of these requirements are easy to see: by feasibility of $S$ no $y$-literal has their complementary literal in $S$ (since they hate each other). Also, a false $x$-literal (currently with utility $10$) would be worse off in $S$ since $t_i,f_i\not\in S$ (by feasibility since $t_i,f_i$ hate a clause) and so the false $x$-literal would obtain at most $5$ in $S$. Thus, false $x$-literals are not part of the weakly blocking $S$. Hence $\pi$ is strict-core-stable. \\
	
	\noindent
	(ii) $\Rightarrow$ (iii): Strict-core-stability implies core-stability. \\
	
	\noindent
	(iii) $\Rightarrow$ (i):  Suppose the game has a core-stable outcome~$\pi$. We show that the formula is true. Recall that every coalition in $\pi$ must be feasible. Write $i\sim j$ for agents $i,j$ that appear in the same coalition in $\pi$ (so $\sim$ is the equivalence relation associated with $\pi$).
	
	First we will find an appropriate assignment for the $x$-variables. Consider variable $x_i$. It cannot be that both $x_i\sim c(x_i)$ and $\overline x_i \sim c(\overline x_i)$ for then by the lemma all clause players are together in $\pi$, which implies $x_i \not\sim t_i, f_i$ (since they hate a clause) and so $x_i$ obtains utility at most $5-10 = -5$, contradicting feasibility. If exactly one of $x_i$ and $\overline x_i$ is together with clause players, say $x_i \sim c(x_i)$, then $\overline x_i$ must obtain utility 0, because feasibility for $f_i$ and $t_i$ implies that they obtain utility at least $30$, which they can only obtain if they are either together with \emph{both} $x_i$ and $\overline x_i$ (but this is impossible because $x_i \not\sim \overline x_i$ by the lemma and our assumption that $\overline x_i \not\sim c(\overline x_i)$) or if they are together with their helper player but not with either $x_i$ or $\overline x_i$.  Thus $\{x_i, \overline x_i, f_i\}$ blocks, a contradiction. 
	The only remaining possibility is that both $x_i\not\sim c(x_i)$ and $\overline x_i\not\sim c(\overline x_i)$. In this case, we can see that core-stability implies either $x_i\sim \overline x_i \sim t_i$ or $x_i\sim \overline x_i \sim f_i$ but not both: having both is infeasible because $t_i$ and $f_i$ hate each other; having neither is unstable because $\{x_i, \overline x_i, f_i\}$ would block. Define the assignment that sets those variables $x_i$ true for which $x_i \sim t_i$, and sets those variables $x_i$ false for which $x_i \sim f_i$.
	
	Soon we will need to know that $\pi$ makes most players not very happy. Indeed, we claim that in $\pi$, clause players $c_k$ obtain utility 30, and $y$-literals obtain utility~0. This is because if any two of these players would obtain strictly more in $\pi$, then case (b) of the lemma applies, so that all clause players are together in $\pi$, and they are together with literals which include an $x_*$-literal. But we already know that $x$-literals are not together with their clause player, a contradiction.
	
	Now take an arbitrary assignment to the $y$-variables, and suppose for a contradiction that under these assignments to the $x$- and $y$-variables the formula $\varphi$ becomes false, so that $\lnot \phi$ becomes true meaning each clause of $\lnot\phi$ is true. Let $S$ be the coalition consisting of all clauses $c_k$, all true $x$-literals, and all true $y$-literals. In $S$, true $x$-literals obtain
	utility $5>4$, 
	$y$-literals obtain positive utility, and clauses obtain utility $\ge 31 = 13+13+5$. Hence everyone in $S$ is strictly better off in $S$ than in $\pi$. Thus $\pi$ is not core-stable, contradiction. Thus, $\varphi$ must be true.
\end{proof}

\section{Conclusions and Open Problems}

We have shown that deciding the existence of a core- or strict-core-stable partition in a given additive hedonic game is $\Sigma_2^p$-complete, and similarly that deciding the existence of a strict-core-stable partition in a Boolean hedonic game is $\Sigma_2^p$-complete. This answers the complexity status of these questions conclusively, and implies that solving them is even harder than solving NP-complete problems (unless the polynomial hierarchy collapses).

The root cause for $\Sigma_2^p$-hardness is that blocking coalitions may be very large. If we consider a solution concept according to which only coalitions of bounded size may block, the \textsc{existence} problem is contained in NP. If the size bound is 2, the problem is identical to the \emph{stable roommates} setting, and in some cases there are known tractability results. For larger size bounds, it is likely that hardness holds. It would be interesting to further explore this variation of the concept of the core.

Our reduction for additive games produced a hardness result that holds even for sparse game that have maximum degree bounded by 10. It would be interesting to decide the complexity of cases where other parameters are small. For example, Ohta et al.\ \citet{ohta2017friends} show that the core remains $\Sigma_2^p$-hard even if agents only assign three different valuation numbers to the other players. Another parameter that may be interesting is the number of agent \emph{types}, that is, the number of different valuation functions appearing in the game (see, e.g., Shrot et al.\ \cite{shrot}). In the framework of graphical hedonic games \citep{peters2016graphical}, it would be interesting to decide the complexity of cases where the underlying graph is planar, or bipartite, or satisfying some other topological constraint. Looking in another direction, our hardness result appears to be the first that applies to sparse games, i.e., games in which agents only assign non-zero valuations to at most a constant number of other agents. It would be interesting to see whether deciding the existence of other solution concepts, such as Nash or individually stable partitions, remains hard for sparse games. More technically, it is likely possible to improve our bound of 10 to a smaller maximum degree.

We also reiterate here a problem posed by Woeginger \citet{WoegingerSurvey}: is strict-core-existence $\Sigma_2^p$-hard even for games based on aversion to enemies? These are additive hedonic games in which $v_i(j)\in\{-\infty,1\}$ for all agents $i,j\in N$. It appears that an entirely different approach is necessary for this setting. Partial progress on this problem is presented by Rey et al.\ \citet{rey2015wonderfully} and by Ohta et al.\ \citet{ohta2017friends}.

\subsubsection*{Acknowledgements.}

I thank the anonymous reviewers for helpful feedback that improved the clarity of presentation, and Lena Schend for useful discussions. I am supported by EPSRC, by ERC under grant number 639945 (ACCORD), and by COST Action IC1205.

\bibliographystyle{splncs03}

\begin{thebibliography}{10}
	\providecommand{\url}[1]{\texttt{#1}}
	\providecommand{\urlprefix}{URL }
	
	\bibitem{aziz2017fractional}
	Aziz, H., Brandl, F., Brandt, F., Harrenstein, P., Olsen, M., Peters, D.:
	Fractional hedonic games  (2017), arXiv:1705.10116 [CS.GT]
	
	\bibitem{ABS11c}
	Aziz, H., Brandt, F., Seedig, H.G.: Computing desirable partitions in
	additively separable hedonic games. Artificial Intelligence  195,  316--334
	(2013)
	
	\bibitem{Aziz2014a}
	Aziz, H., Brandt, F., Harrenstein, P.: Fractional hedonic games. In:
	Proceedings of the 13th International Conference on Autonomous Agents and
	Multiagent Systems (AAMAS). pp. 5--12 (2014)
	
	\bibitem{Booleanhedonicgames2014}
	Aziz, H., Harrenstein, P., Lang, J., Wooldridge, M.: Boolean hedonic games. In:
	Proceedings of the 15th International Conference on Principles of Knowledge
	Representation and Reasoning (KR). pp. 166--175 (2016)
	
	\bibitem{azizsavani}
	Aziz, H., Savani, R.: Hedonic games. In: Brandt, F., Conitzer, V., Endriss, U.,
	Lang, J., Procaccia, A.D. (eds.) Handbook of Computational Social Choice,
	chap.~15. Cambridge University Press (2016)
	
	\bibitem{Ballester2004}
	Ballester, C.: {NP}-completeness in hedonic games. Games and {Economic}
	{Behavior}  49(1),  1--30 (2004)
	
	\bibitem{Banerjee2001}
	Banerjee, S., Konishi, H., S\"{o}nmez, T.: Core in a simple coalition formation
	game. Social {Choice} and {Welfare}  18(1),  135--153 (2001)
	
	\bibitem{Berman2003}
	Berman, P., Karpinski, M., Scott, A.D.: Approximation hardness of short
	symmetric instances of {MAX}-3{SAT}. Tech. Rep. ECCC TR03-049 (2003)
	
	\bibitem{Bogomolnaia2002}
	Bogomolnaia, A., Jackson, M.O.: The stability of hedonic coalition structures.
	Games and {Economic} {Behavior}  38(2),  201--230 (2002)
	
	\bibitem{Cechlarova2004}
	Cechl\'{a}rov\'{a}, K., Hajdukov\'{a}, J.: Stable partitions with $\mathcal
	{W}$-preferences. Discrete {Applied} {Mathematics}  138(3),  333--347 (2004)
	
	\bibitem{Dimitrov2006}
	Dimitrov, D., Borm, P., Hendrickx, R., Sung, S.C.: Simple priorities and core
	stability in hedonic games. Social {Choice} and {Welfare}  26(2),  421--433
	(2006)
	
	\bibitem{Elkind2009}
	Elkind, E., Wooldridge, M.: Hedonic coalition nets. In: Proceedings of the 8th
	International Conference on Autonomous Agents and Multiagent Systems (AAMAS).
	pp. 417--424 (2009)
	
	\bibitem{ko1995complexity}
	Ko, K.I., Lin, C.L.: On the complexity of min-max optimization problems and
	their approximation. In: Minimax and Applications, pp. 219--239. Springer
	(1995)
	
	\bibitem{malizia2007infeasibility}
	Malizia, E., Palopoli, L., Scarcello, F.: Infeasibility certificates and the
	complexity of the core in coalitional games. In: Proceedings of the 20th
	International Joint Conference on Artificial Intelligence (IJCAI). pp.
	1402--1407 (2007)
	
	\bibitem{ohta2017friends}
	Ohta, K., Barrot, N., Ismaili, A., Sakurai, Y., Yokoo, M.: Core stability in
	hedonic games among friends and enemies: Impact of neutrals. In: Proceedings
	of the 26th International Joint Conference on Artificial Intelligence (IJCAI)
	(2017)
	
	\bibitem{peters2016dichotomous}
	Peters, D.: Complexity of hedonic games with dichotomous preferences. In:
	Proceedings of the 30th AAAI Conference on Artificial Intelligence (AAAI).
	pp. 579--585 (2016)
	
	\bibitem{peters2016graphical}
	Peters, D.: Graphical hedonic games of bounded treewidth. In: Proceedings of
	the 30th AAAI Conference on Artificial Intelligence (AAAI). pp. 586--593
	(2016)
	
	\bibitem{Peters2015}
	Peters, D., Elkind, E.: Simple causes of complexity in hedonic games. In:
	Proceedings of the 24th International Joint Conference on Artificial
	Intelligence (IJCAI). pp. 617--623 (2015)
	
	\bibitem{rey2015wonderfully}
	Rey, A., Rothe, J., Schadrack, H., Schend, L.: Toward the complexity of the
	existence of wonderfully stable partitions and strictly core stable coalition
	structures in enemy-oriented hedonic games. Annals of Mathematics and
	Artificial Intelligence pp. 1--17 (2015)
	
	\bibitem{shrot}
	Shrot, T., Aumann, Y., Kraus, S.: On agent types in coalition formation
	problems. In: Proceedings of the 9th International Conference on Autonomous
	Agents and Multiagent Systems (AAMAS). pp. 757--764 (2010)
	
	\bibitem{stockmeyer1976polynomial}
	Stockmeyer, L.J.: The polynomial-time hierarchy. Theoretical Computer Science
	3(1),  1--22 (1976)
	
	\bibitem{Sung2010}
	Sung, S.C., Dimitrov, D.: Computational complexity in additive hedonic games.
	European {Journal} of {Operational} {Research}  203(3),  635--639 (2010)
	
	\bibitem{WoegingerSurvey}
	Woeginger, G.J.: Core stability in hedonic coalition formation. In: SOFSEM
	2013: Theory and Practice of Computer Science, pp. 33--50. Springer (2013)
	
	\bibitem{Woeginger2013}
	Woeginger, G.J.: A hardness result for core stability in additive hedonic
	games. Mathematical {Social} {Sciences}  65(2),  101--104 (2013)
	
\end{thebibliography}

\end{document}